\documentclass[letterpaper, 10 pt, conference]{ieeeconf}  

\IEEEoverridecommandlockouts                              

\usepackage{multirow}

\usepackage{amsthm}
\usepackage{amsmath, comment} 
\usepackage{amssymb}  
\usepackage{algorithm} 
\usepackage{algorithmic}

\usepackage{xcolor}
\newtheorem{lemma}{Lemma}
\newtheorem{thm}{Theorem}

\newtheorem{coro}{Corollary}

\newtheorem{prop}{Proposition}

\theoremstyle{definition}

\theoremstyle{remark}
\newtheorem{rem}{Remark}
\usepackage{siunitx}

\usepackage{caption}
\usepackage{subcaption}

\usepackage{graphicx}

\usepackage{stackengine,cite}
\def\delequal{\mathrel{\ensurestackMath{\stackon[1pt]{=}{\scriptstyle\Delta}}}}

\title{\LARGE\bf A (Strongly) Connected Weighted Graph is Uniformly Detectable based on any Output Node}

\author{Uduak Inyang-Udoh, Michael Shanks, and Neera Jain
\thanks{*This work is supported by the U.S. Office of Naval Research Thermal Science and Engineering Program under contract number N00014-21-1-2352.}
\thanks{The authors are with the School of Mechanical Engineering, Purdue University, West Lafayette, Indiana 47907 USA
        {\tt\small uinyangu@purdue.edu, shanks5@purdue.edu,
        neerajain@purdue.edu}}
}

\begin{document}

\maketitle
\thispagestyle{empty}
\pagestyle{empty}

\begin{abstract}
Many dynamical systems, including thermal, fluid, and multi-agent systems, can be represented as weighted graphs. In this paper we consider whether the unstable states of such systems can be observed from limited discrete-time measurement, that is, whether the discrete formulation of system is detectable. We establish that if the associated graph is fully connected, then a linear time invariant system is detectable by measuring any state. Further, we show that a parameter-varying or time-varying system remains uniformly detectable for a reasonable approximation of its state transition matrix, if the underlying graph structure is maintained.

\end{abstract}


\section{Introduction}
A key challenge in systems control is ascertaining whether the states of a system that cannot be measured can be retrieved, or \emph {observed} \cite{chen_linear_2014}. When a system is unobservable, or it is difficult to establish \textit{observability}, it is often satisfactory that the system is \textit{detectable}, that is, (unmeasured) states that cannot be retrieved are stable \cite{wonham_linear_2012,anderson_detectability_1981}. A dual problem relates to whether states that cannot directly receive an input can be \emph{controlled} by input(s) from other states, or whether uncontrollable states are stable (\textit{stabilizable}). The question about these system properties are often analyzed on a system-by-system basis. However, many dynamical systems---including thermal, fluid and multi-agent systems---can be modeled as graphs \cite{Aksland,Inyang-Udoh2021, olfati-saber_consensus_2007, mesbahi_graph_2010}. As such, by using the tools of graph theory, we can more generally analyze the control properties of these systems. This is particularly relevant in time-varying systems or nonlinear systems where the systems' parameters may change with time but the underlying graph structure remains intact \cite{liu_controllability_2011}. 

The question of whether a given graph is observable or controllable in the classical sense \cite{kalman_general_1960, hautus_stabilization_1970} has been studied for different graph topologies \cite{montanari_observability_2020}. For example, a path graph has been shown to be observable from either end node and from certain other interior nodes depending on its spectrum \cite{parlangeli_observability_2010}. On the other hand, a cycle graph may only be observed from a pair of nodes, and an unweighted complete graph is uncontrollable and unobservable \cite{tanner_controllability_2004} from any single node. Further, it has been demonstrated that for an undirected unweighted grid graph (or lattice), the observation nodes (that is, nodes from which the graph can be observed) depend on the spectrum of its constituent path graphs \cite{notarstefano_observability_2011, notarstefano_controllability_2013}. Indeed, in general, a graph obtained by \emph{joining} several graphs is observable from any single node if the eigenvectors of all constituent graphs (adjacency matrices) are mutually distinct and \emph{main} \cite{cvetkovic_controllable_2011,andelic_laplacian_2020, rowlinson_main_2007, cvetkovic1980spectra}.

The literature cited here only provides generalizations about the observability or controllability of a dynamic system if the system's graph has a simple topology (e.g.\ a path graph), or, for slightly more complex topologies, has unweighted undirected edges  (e.g.\ a lattice). However, most system graphs are arbitrarily weighted or complex \cite{liu_controllability_2011}. To address this, the concepts of \emph{structural observability} (and \emph{structural controllability}) have been adopted to account for when a certain realization of the graph structure (combination of edge weights) is observable (or controllable) \cite{lin_structural_1974, shields_structural_1976, dion_generic_2003, li_resilient_2019}. A graph that is structurally observable is observable for most weight combinations. However, this is not so if certain weights are fixed \cite{lin_structural_1974}.  

Given the difficultly in establishing observability or controllabilty for a system's graph representation, we note that the properties of detectability and stabilizability  are sufficient to establish the existence of a stable observer and controller \cite{anderson_detectability_1981}. Thus, in this paper, we generalize detectability for systems whose graphs are strongly connected. Such systems include those whose dynamical equations obey conservation laws, such as those of mass and energy. First, we consider linear time-invariant (LTI) systems with discrete-time measurements whose graph Laplacians are weighted and demonstrate that the systems are detectable provided the output matrix does not have a zero row sum. We then show that this result applies to parameter-varying systems as long as the graph structure remains fully connected. 

The paper is organized as follows. In Section \ref{sec-prob_statement}, we describe motivating problems that make graph detectability pertinent. Concepts used in the rest of the paper are introduced in Section \ref{sec-preliminaries}. We present the main results in Section \ref{sec-analysis} and draw conclusions in Section \ref{sec-conclusion}.  

\section{Problem Motivation} \label{sec-prob_statement}
Many autonomous systems are governed by an equation of the form
\begin{align}
    \frac{\partial^{n_l}\psi}{\partial^{n_l}t} = \nabla \cdot \left(c(\psi)\nabla^{n_r} \psi\right),
\end{align}
where $\psi$ is the state of the system;  $\nabla$ is the differential operator del; and $c(\psi)$ is a parameter of appropriate size, equal to constant $c$ for a linear system. For example, 
\begin{enumerate}
    \item if $n_l = 1$, $n_r = 0$, and $c(\psi)$ is a velocity vector field, the equation represents an advection which govern the dynamics in many fluid transport systems \cite{BRIO2010175};
    \item if $n_l = 1$, $n_r = 1$, and $c(\psi)$ denotes the diffusion coefficient, the diffusion equation, which is essential in diverse kinds of transport systems, results \cite{crank_mathematics_1979}; and
    \item if $n_l = 2$, $n_r = 1$, and $c(\psi)$ is some constant, the wave propagation equation emerges \cite{elmore2012physics}.
\end{enumerate}
These partial differential equations (PDEs) give rise to graphs 
\begin{itemize}
    \item when the domain over which the PDEs apply is discretized for finite difference analysis, or
    \item in lumped parameter modeling of systems by the PDEs on a macroscopic scale.
\end{itemize}  
Based on the graph representation, we seek to assess if the system is detectable (or stabilizable, by duality) given limited output measurement because the detectability of a system assures that we can develop a bounded or converging Riccati-equation based estimator such as a Kalman filter \cite{kalman_new_1960} or a State-Dependent Riccati-Equation (SDRE) filter (for a nonlinear system) \cite{anderson_detectability_1981, beikzadeh_stability_2009, beikzadeh_exponential_2012}.
\section{Preliminaries} \label{sec-preliminaries}
In this section, we define basic concepts that undergird the main result. The following notations are used: $\left\| \cdot \right\| $  denotes the vector norm or the induced matrix norm; inequality operators such as $\geq$ or $>$ used with matrices apply componentwise.
\subsection{Graphs}
A graph is a pair $\mathcal G = (\mathcal V,\mathcal E)$ where $\mathcal V$ is is a finite set of $n$ nodes or vertices  $\{v_1,\dots,v_n\}$,  and $\mathcal E \subseteq \{1,\dots,n\} \times \{1,\dots,n\}$ is a set of edges. The pair $(i,j)$ denotes the edge that links vertex $v_i$ to $v_j$. A \textit{simple} or \textit{strict}  graph is one that has no  edges of the form $(i,i)$ (self-loops) or multiple edges between $v_i$ and $v_j$ ($\mathcal E \subset \{1,\dots,n\} \times \{1,\dots,n\}$). Unless stated, graphs are usually assumed to be simple.  We define the following graph-related terms used in this paper. 

\textit{Definition 1 (Undirected and Directed Graphs):} A graph is \textit{undirected}, or \textit{symmetric}, when $(i,j) \in \mathcal E \Leftrightarrow (j,i) \in \mathcal E$. Otherwise, (if the edges are oriented), it is \textit{directed} or simply regarded as a \textit{digraph}.

\textit{Definition 2 (Connectedness):} An undirected graph is \textit{connected} if there exists a path (set of edges) between any two nodes. If the graph is directed, that is, there is a path of oriented edges between the two nodes, it is \textit{strongly connected}.  

\textit{Definition 3 (Weighted Graph):} A graph is weighted if a factor $w_{i,j} > 0$ (called \textit{weight}) is assigned to any of its edges $(i,j)$. In an unweighted graph $w_{i,j} = 1 ~\forall ~ (i,j)$ if the edge $(i,j)$ exists; otherwise (for a weighted or unweighted graph) $w_{i,j} = 0$.

\textit{Definintion 4 (Adjacency and Laplacian Matrices):} A (weighted) graph is fully described by the adjacency matrix, $\Lambda \in \mathbb{R}^{n \times n}$ where
\begin{align}
\Lambda_{i,j} = 
\begin{cases}w_{i,j} & \text{if } (i,j) \in E \\
0& \text{otherwise}.  \end{cases}
\end{align} 
Consider the diagonal matrix $D$ whose diagonal entry $D_{i,i}=\sum_{j=1}^{n}\Lambda_{i,j}$. $D$ contains the  total egdes' weight that is incident to each vertex, and is termed the \textit{degree matrix}. We define the graph Laplacian $L = D-\Lambda$, that is,
\begin{align}
L_{i,j} = 
\begin{cases}\text{deg}(v_i) & \text{if } i=j \\
-w_{i,j} & \text{if } i\neq j \text{ and } (i,j) \in E \\
0& \text{otherwise},  \end{cases}\label{eq:Laplacian}
\end{align}
where deg$(v_i)$ is the degree of vertex $v_i$: $\sum_{j=1}^{n} D_{i,j}$.
\subsection{Irreducible Matrices} 

\textit{Definition 5:} A matrix $A \in \mathbb{R}^{n \times n}$ is irreducible if there does not exist a permutation matrix $P$ such that the matrix $J=P A P^{T}$ is in upper block triangular form:
\begin{align}
J=P A P^{T}=\left(\begin{array}{cc}
J_{11} & J_{12} \\
O & J_{22}
\end{array}\right),
\end{align}
where the diagonal blocks $J_{11}$, $ J_{12}$ are non-trivial square matrices. Otherwise, it is reducible.

The concept of irreducibility is  strongly associated with graph connectivity as discussed in the following lemma. 
\begin{lemma}[\cite{meyer2000matrix}]
A matrix $A \in \mathbb{R}^{n \times n}$, corresponding to the adjacency matrix of some graph $\mathcal{G}_A$, is irreducible if and only if the $\mathcal{G}_A$ is (strongly) connected. 
\end{lemma}
\begin{coro}
    The Laplacian $L$ of a graph $\mathcal G_A$ is irreducible if it is (strongly) connected.
\end{coro}

\begin{proof}
    The Laplacian associated with the graph $\mathcal G_A$ is $L=D-A$. For any permutation matrix $P$,
\begin{align}
    P L P^{T} = P D P^{T} - P A P^{T}
\end{align}
is not upper block triangular since $P D P^{T}$ is diagonal and, for a (strongly) connected graph, $P A P^{T}$ is never upper block triangular.
\end{proof}
Irreducible matrices have ``elegant" properties \cite{meyer2000matrix}; one such property is highlighted in the following lemma.
\begin{lemma}[Perron-Frobenius \cite{meyer2000matrix}]
    If $A \in \mathbb{R}^{n \times n}$ is an irreducible nonnegative matrix $( A \geq 0)$, then $ (I_n+A)^{n-1}$ is a positive matrix  $\left( (I_n+A)^{n-1} > 0\right)$.
\end{lemma}

\begin{coro}
    If $A \in \mathbb{R}^{n \times n} \geq 0$ is an irreducible matrix, then $ e^{A} > 0$.
\end{coro}
\begin{proof}
    Note that $(I_n+A)^{n-1} = \overset{n-1}{\underset{{k = 0}}{\sum}} {n-1 \choose k} A^k $. Since $A\geq 0$,  $\overset{n-1}{\underset{{k = 0}}{\sum}} {n-1 \choose k} A^k > 0 \implies \overset{n-1}{\underset{{k = 0}}{\sum}} \frac{1}{k!} A^k > 0$.  Then clearly,
\begin{align}
    e^{A} = \sum^{\infty}_{k = 0}\frac{A^k}{k!} \geq \sum^{n-1}_{k = 0}\frac{A^k}{k!}  >0.
\end{align}
\end{proof}

\subsection{Detectability}
Consider a discrete state-space system
\begin{subequations}\label{eq:sys_dis}
\begin{align}
        x_{\mathsf{k}+1} &= A_{\mathsf{k}} x_{\mathsf{k}}\\
    y_{\mathsf{k}} &= C_{\mathsf{k}} x_{\mathsf{k}} 
\end{align}
\end{subequations}
where $x_{\mathsf{k}} \in \mathbb{R}^{n}$ is the state vector, $y_{\mathsf{k}} \in \mathbb{R}^{m}$ is the output vector; $A_{\mathsf{k}} \in \mathbb{R}^{n\times n}$  and $C_{\mathsf{k}} \in \mathbb{R}^{m \times n}$ are the system state and output matrices respectively. Further, for two positive successive integers $\mathsf{k}, \mathsf{l}$ ($\mathsf{k}\geq \mathsf{l}$), let the \textit{state transition matrix} be defined as $\Phi_{\mathsf{l}|\mathsf{k}} = \Phi_{\mathsf{l}|\mathsf{l}-1}\Phi_{\mathsf{l}-1|\mathsf{k}}$ where $\Phi_{\mathsf{k}|\mathsf{k}} = I_n$ and $\Phi_{\mathsf{k}+1|\mathsf{k}} = A_{\mathsf{k}}$. Then the system is said to be detectable if the following definition holds. 

\textit{Definition 6 (Uniform Detectability} \cite{anderson_detectability_1981}\textit{):} The pair $(A_{\mathsf{k}},C_{\mathsf{k}})$ is uniformly detectable if there exist integers $p$, $q$, and some $0 \leq a<1$, $b>0$ such that whenever
\begin{align}
\left\|\Phi_{\mathsf{k}+\mathsf{p}|\mathsf{k}}\zeta\right\| & \geq  a\left\|\zeta\right\|
\end{align}
for any $\zeta \in \mathbb{R}^{n}$ and $\mathsf{k}$, then
\begin{align}
\zeta^TW_{\mathsf{k}+\mathsf{q}|\mathsf{k}}\zeta \geq b\zeta^T\zeta,
\end{align}
where $W_{\mathsf{k}+\mathsf{q}|\mathsf{k}}$  is the discrete-time observability gramian given by
\begin{align}\label{eq:gramian_disc}
   W_{\mathsf{k}+\mathsf{q}|\mathsf{k}} \delequal \sum^{k+q}_{i = k}\Phi_{\mathsf{i}|\mathsf{k}}^{T}C_{\mathsf{k}}^TC_{\mathsf{k}}\Phi_{\mathsf{i}|\mathsf{k}}.
\end{align}

\section{Main Result} \label{sec-analysis}
In this section, we show that the structure of the state transition matrix that corresponds to a connected undirected graph or a strongly connected digraph forms a dectectable pair with any nonnegative and nonzero output matrix.


\begin{figure}[h!]
\centering
    \includegraphics[clip, trim=0cm 0cm 0cm 0cm, width=0.35\textwidth]{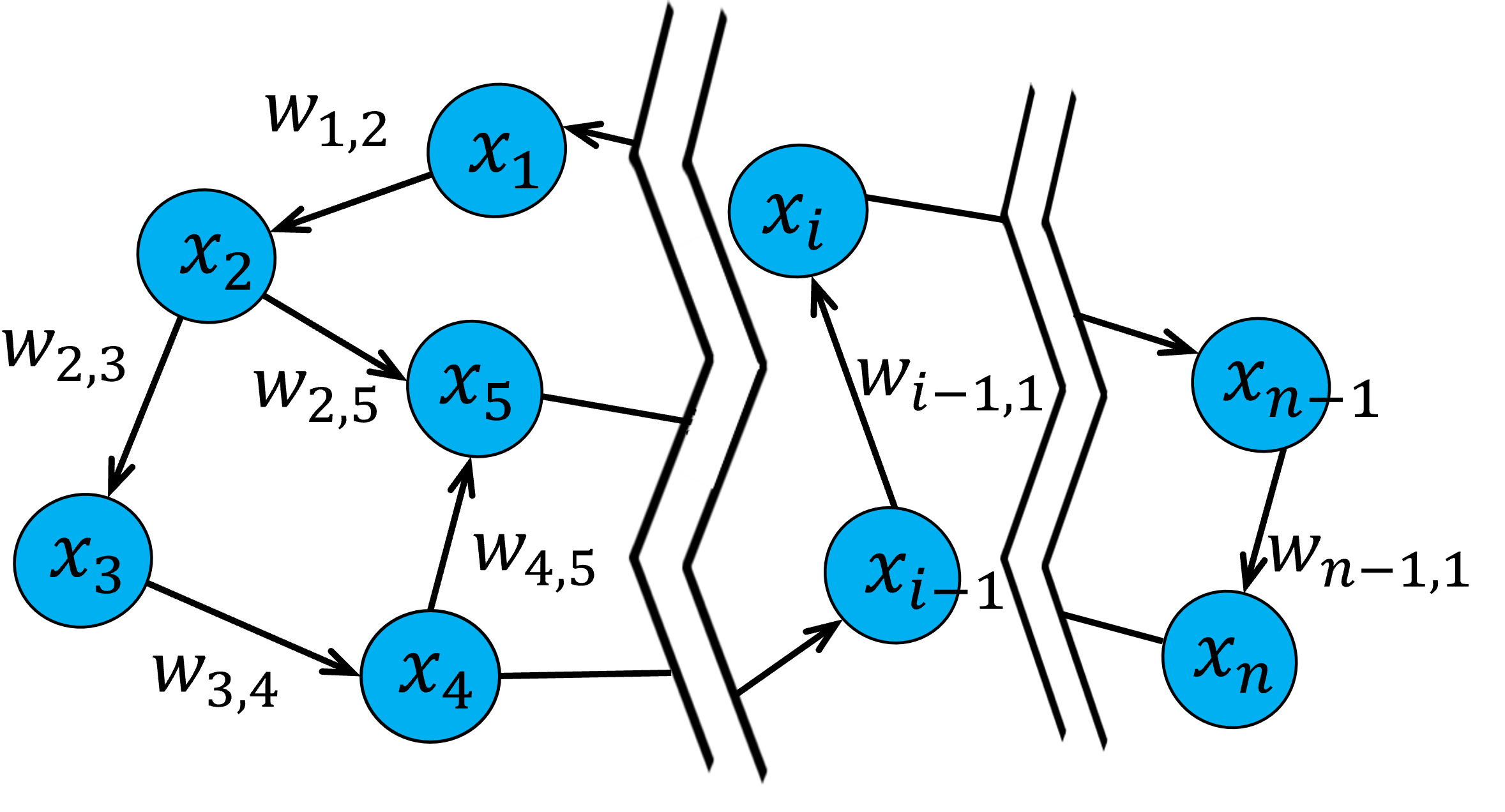}
   \caption{Graph with nodes $x_i$, $i\in{1,2,\dots,n}$ and edge weights $w_{i,j}$ where $j \in \mathcal{N}(i)$, that is, $j$ is in the neighborhood of $i$.}
      \label{pic:weight}
\end{figure}

Consider the system shown in Fig. \ref{pic:weight} where the dynamics of each node $x_i$ are given by
\begin{align}
    \dot{x}_i &=\sum_{j \in \mathcal{N}(i)} w_{i,j}\left(x_j - x_i\right) ~~ i\in\{1, \dots, n\}, \label{eq:sys_graph}
\end{align}
and $w_{i,j}$ is the weight on the vertex $(i,j)$ (if no vertex exists, $w_{i,j}=0$). This system corresponds to a weighted graph and can be represented in state-space form as
\begin{align}
    \dot{x}&= -Lx, \label{eq:sys}
\end{align}
where $L$ is the Laplacian as defined in \eqref{eq:Laplacian}. For a connected undirected graph or strongly connected digraph, the irreducibility of $L$ leads to the following result.
\begin{lemma}
    Let $L$ be the Laplacian of a (strongly) connected graph. Then $e^{-L} > 0$. 
\end{lemma}
\begin{proof}
    Since all non-diagonal terms of $-L$ are non-negative, $\exists \enspace \gamma>0$ such that $A_\gamma = -L + \gamma I_n \geq 0$. From the proof of Corollary 1, it is clear that if $L$ is irreducible, $A_\gamma$ is also irreducible. Hence, from Corollary 2, $e^{A_\gamma}$ is strictly positive. Therefore, 
$$
e^{-L} = e^{A_{\gamma}-\gamma I_n} = e^{-\gamma}e^{A_{\gamma}}
$$
is also strictly positive.
\end{proof}

In addition to the strict positiveness of $e^{-L}$, the characteristic space associated with $L$ places further bounds on the elements and norm of $e^{-L}$ as stated in the following proposition and the ensuing corollary.
\begin{prop}
    Let $L$ be the Laplacian of a (strongly) connected graph. Then $e^{-L}$ is a right stochastic matrix whose entries are positive real numbers less than $1$, that is, $0 <e^{-L} < 1$. 
\end{prop}

\begin{proof}
    We note that that the vector $v = \mathbf{1}_n$ where $\mathbf{1}_n = \begin{bmatrix}
    1 & 
    \cdots & 1
\end{bmatrix}^T \in \mathbb{R}^n$ is a right eigenvector of $L$ with a $0$ eigenvalue:
\begin{align}
Lv = L\cdot \mathbf{1}_n= \begin{bmatrix}
    \sum^n\limits_{j=2}L_{1,j}- \sum^n\limits_{j=2}L_{1,j}
    \\ \vdots \\
    \sum^n\limits_{\substack{j=1 \\ j\neq i}}L_{i,j}- \sum^n\limits_{\substack{j=1 \\ j\neq i}}L_{i,j}\\ \vdots \\
    \sum^{n-1}\limits_{j=1}L_{n,j} -\sum^{n-1}\limits_{j=1}L_{n,j}
\end{bmatrix} = 0 \cdot \mathbf{1}_n
\end{align}
since for a (strongly) connected graph there exists some $j$ such that $L_{i,j}>0$ $\forall i \neq j$. It follows that $v = \mathbf{1}_n$ must also be an eigenvector of $e^{-L}$ with eigenvalue $e^0 = 1$. Hence
\begin{align}
    e^{-L} \cdot \mathbf{1}_n = \mathbf{1}_n ,
\end{align}
or
\begin{align}
    \sum^n\limits_{j=1}e^{-L}_{i,j}= 1 ~~\forall i .
\end{align}
Recall from Lemma 3 that $e^{-L}_{i,j} > 0$. Thus, for a graph with $n\geq 2$,
\begin{align}
    0 < e^{-L}_{i,j} < 1.
\end{align}
\end{proof}
\begin{coro}
    The induced $\infty-$norm of $e^{-L}$ is unique, that is, 
vector $v= \mathbf{1}_n$ is the unique solution of 
\begin{align}
   \max_{\left\|v\right\|_\infty = 1} \left\|e^{-L}v\right\|_\infty.
\end{align}
\end{coro}
\begin{proof}
    Since $e^{-L}$ is right stochastic, if  $v= \mathbf{1}_n$, 
\begin{align}
 \left\|e^{-L}v\right\|_\infty = \left\|e^{-L} \mathbf{1}_n\right\|_\infty = \left\|\mathbf{1}_n\right\|_\infty =\left\|v\right\|_\infty = 1.
\end{align}
If $v\neq \mathbf{1}_n$ while $\left\|v\right\|_\infty = 1$, since  $0 <e^{-L} < 1$, 
\begin{align}
 \left\|e^{-L}v\right\|_\infty <
 \left\|\mathbf{1}_n\right\|_\infty =\left\|v\right\|_\infty = 1.
\end{align}
Hence, $v= \mathbf{1}_n$ is the unique argument of $\max_{\left\|v\right\|_\infty = 1} \left\|e^{-L}v\right\|_\infty$.
\end{proof}
We now consider the implication of the the uniqueness of the $\infty-$norm of $e^{-L}$ on the system's detectability both (i) when the system's Laplacian is time-invariant and (ii) when it is parameter-varying.  In the time-invariant case, the uniqueness of the $\infty-$norm of $e^{-L}$ immediately implies that any unobservable mode of the system is stable. For approximations of the state-transition matrix in the parameter-varying case, the same argument holds.
\subsection{Linear Time-Invariant System}
Suppose the system in \eqref{eq:sys} has a piecewise constant input $u_\mathsf{{k}} \in \mathbb{R}^{r}$ and output $y_\mathsf{{k}}\in \mathbb{R}^{m}$ over time interval $\Delta t = t_{\mathsf{k+1}} - t_{\mathsf{k}}$, the discrete-time formulation of the system may be written as
\begin{subequations}\label{eq:sys_disc}
\begin{align}
        x_{\mathsf{{k}}+1} &= e^{-L\Delta t}x_\mathsf{{k}} +Bu_\mathsf{k}\\
    y_\mathsf{{k}} &= Cx_\mathsf{{k}}
\end{align}
\end{subequations}
where $B$ and $C$ are the input and output matrices respectively. The structure and bounds on $e^{-L}$ enable us to make the following generalization on the detectability of the system.

\begin{thm}
    Let $L$ be the Laplacian of a (strongly) connected graph, then the pair 
\begin{enumerate}
    \item $(e^{-L\Delta t}, C)$ is uniformly detectable if $C$ has at least one non-zero row sum;
    \item $(e^{-L\Delta t}, B)$ is uniformly stabilizable if $B$ has at least one non-zero column sum.
\end{enumerate}
\end{thm}
\begin{proof}
    1) Let $\Phi_{\mathsf{k}+\mathsf{p}|\mathsf{{k}}} = e^{-L\mathsf{p}\Delta t} $. Note that $L\mathsf{p}\Delta t$ is simply a scalar multiple of $L$ and hence is still a (strongly) connected graph Laplacian. Following Definition 6, we need to evaluate $v$ such that $\nexists ~ a \in [0,1)$ that satisfies 
\begin{align}
    \frac{\left\|\Phi_{\mathsf{k}+\mathsf{p}|\mathsf{k}}v\right\|}{\left\|v\right\|} < a , \label{eq:detect_norm}
\end{align}
 and show that for such $v$,
 \begin{align}
v^TW_{\mathsf{k}+\mathsf{q}|\mathsf{k}}v \geq bv^Tv ~~ (b>0)\label{eq:detect_2norm}
\end{align}
for any $q$. Observe that maximum of the left-hand-side of \eqref{eq:detect_norm} is, by definition, the induced norm of $\Phi_{\mathsf{k}+\mathsf{p}|\mathsf{k}}$:
\begin{align}
     \left\|\Phi_{\mathsf{k}+\mathsf{p}|\mathsf{k}} \right\|= \max_{v\neq 0} \frac{\left\|\Phi_{\mathsf{k}+\mathsf{p}|\mathsf{k}}v\right\|}{\left\|v\right\|} =  \max_{\left\|v\right\| = 1} \left\|\Phi_{\mathsf{k}+\mathsf{p}|\mathsf{k}}v\right\|.
\end{align}
Consider the infinity norm of $\Phi_{\mathsf{k}+1|\mathsf{k}} = e^{-L\Delta t}$ $(\mathsf{p} = 1)$. From Corollary 3, we have that iff $v= \mathbf{1}_n$,
\begin{align}
    \frac{\left\|\Phi_{\mathsf{k}+\mathsf{p}|\mathsf{k}}v\right\|_\infty}{\left\|v\right\|_\infty} = \left\|\Phi_{\mathsf{k}+\mathsf{p}|\mathsf{k}} \right\|_\infty \nless  a ~~~\text{for } a \in [0,1).
\end{align}
Thus, we only need to show that \eqref{eq:detect_2norm} is satisfied with $v= \mathbf{1}_n$. Since $C$ has at least one non-zero row sum, $\exists ~ b$ such that
\begin{align}
v^TW_{\mathsf{k}+\mathsf{q}|\mathsf{k}}v  &= v^T\sum^{k+q}_{i = k}\Phi_{\mathsf{i}|\mathsf{k}}^{T}C_{\mathsf{k}}^TC_{\mathsf{k}}\Phi_{\mathsf{i}|\mathsf{k}}v \nonumber \\&= q \mathbf{1}_n^TC^TC\mathbf{1}_n \geq b>0.
\end{align}

2) This follows immediately from the duality of detectability with stabilizability \cite{anderson_detectability_1981}.
\end{proof}
In a graph system, it is typical to take measurements directly from the nodes, that is, the output is a vector of select system (node) states, as described in the following definition. 

\textit{Definition 7:} A LTI system with state $x \in \mathbb{R}^{n}$ and output $y = Cx \in \mathbb{R}^{m}$ has \textit{state output} if $C\in \mathbb{R}^{m\times n}$ is a binary output matrix where $C_{i,j} = 1$ only if $y_{i} = x_{j}$ for $i \in \{1,\dots m\}$, $j \in \{1,\dots n\}$.

The output matrix $C$ in such a system conforms to Theorem 1 and hence, we have the following corollary.
\begin{coro}
    The discrete-time system in \eqref{eq:sys_disc} is uniformly detectable if it has state output.
\end{coro}
The proof is evident from Definition 7.


\subsection{Parameter-Varying System}

Instead of being constant, suppose the weights in the dynamical system given in \eqref{eq:sys_graph} are functions of some time-dependent parameter $\rho(t)$.  Then the system given in \eqref{eq:sys} becomes Linear Parameter-Varying (LPV), or quasi-LPV if $\rho$ is a function of the state $x$ itself (that is, the system is nonlinear) \cite{huang_nonlinear_nodate}. Hence, we can write, without loss of generality, that
\begin{align}
    \dot{x}&= -L\left(\rho(x,t)\right)x. \label{eq:qlpv_sys}
\end{align}
In general, there are no closed form solutions for \eqref{eq:qlpv_sys}, and the state transition matrix is usually evaluated numerically. 
Over the time interval $[t_0, t_f]$, the state transition matrix may be expressed as \cite{rugh1993linear}:
 \begin{align}
    \Phi(t_0, t_f) &= \Phi(t_0, t_1) \Phi(t_1, t_2)\dots \Phi(t_{f-1}, t_f)\label{eq:st_trans}
\end{align}
Assume that each $\delta t_i =  t_{i+1} - t_{i}$, $i \in \{1, 2, \dots, f \}$ is sufficiently small, and $L\left(\rho(x,t)\right) \equiv L_i$ for $t \in [t_{i}, t_{i+1}]$. Then, since $e^{-L_i \delta t_i}$ is positive and right stochastic (Corollary 3), the state transition matrix 
\begin{align}
    \Phi(t_0, t_f) &= \prod_{i = 0}^f e^{-L_i \delta t_i},\label{eq:st_trans_approx}
\end{align}
is also positive and right stochastic. Based on Theorem 1, the system (in \eqref{eq:qlpv_sys}) will, therefore, also be detectable.
\section{Conclusion}\label{sec-conclusion}
In this paper, we have shown that if a system can be represented as an undirected fully connected or a directed strongly connected graph, then the system is detectable if measurement is made from any node. This is because the only unstable mode is observable from any node measurement. We have demonstrated this result for a linear time-invariant system and shown that by approximate solution for the state transition matrix, the same result applies to a linear parameter-varying system. The implication is that for all such systems, a stable Kalman filter or SDRE filter can be synthesized.





\bibliographystyle{IEEEtran}
\bibliography{Detectability_ref}

\end{document}